\documentclass[11pt]{article} 

\usepackage[utf8]{inputenc} 
\usepackage{geometry} 
\usepackage{fancyhdr} 
\usepackage{amsthm}

\usepackage[colorlinks]{hyperref} 

\usepackage[utf8]{inputenc} 
\usepackage[T1]{fontenc}
\usepackage{url}
\usepackage{ifthen}
\usepackage{cite}
\usepackage[cmex10]{amsmath} 
\usepackage{graphicx}
\usepackage{amssymb}
\usepackage{xcolor}

\usepackage{algorithm}
\usepackage{algorithmic}
\usepackage{amsmath}

\newtheorem{theorem}{Theorem}

\newtheorem{lemma}{Lemma}

\newtheorem{definition}{Definition}

\newcommand{\bu}{\boldsymbol{u}}
\newcommand{\ba}{\boldsymbol{\alpha}}
\newcommand{\diag}{\ensuremath{\operatorname{diag}}}

\begin{document}

\begin{center}

	{\bf{\LARGE{Graph-Based Ascent Algorithms for Function Maximization}}}

	\vspace*{.25in}

	\begin{tabular}{ccccc}
		{\large{Muni Sreenivas Pydi$^*$}} & \hspace*{.75in} & {\large{Varun Jog$^*$}} & \hspace*{.75in} & {\large{Po-Ling Loh$^{*\dagger}$}}\\
		{\large{\texttt{pydi@wisc.edu}}} & \hspace*{.75in} & {\large{\texttt{vjog@ece.wisc.edu}}} & \hspace*{.75in} & {\large{\texttt{loh@ece.wisc.edu}}}
			\end{tabular}
\begin{center}
Department of Electrical \& Computer Engineering$^*$\\
Department of Statistics$^\dagger$\\
University of Wisconsin - Madison\\
1415 Engineering Drive\\
Madison, WI 53706
\end{center}

	\vspace*{.2in}

	February 2018

	\vspace*{.2in}

\end{center}

\begin{abstract}
We study the problem of finding the maximum of a function defined on the nodes of a connected graph. The goal is to identify a node where the function obtains its maximum. We focus on local iterative algorithms, which traverse the nodes of the graph along a path, and the next iterate is chosen from the neighbors of the current iterate with probability distribution determined by the function values at the current iterate and its neighbors. We study two algorithms corresponding to a Metropolis-Hastings random walk with different transition kernels: (i) The first algorithm is an exponentially weighted random walk governed by a parameter $\gamma$. (ii) The second algorithm is defined with respect to the graph Laplacian and a smoothness parameter $k$. We derive convergence rates for the two algorithms in terms of total variation distance and hitting times. We also provide simulations showing the relative convergence rates of our algorithms in comparison to an unbiased random walk, as a function of the smoothness of the graph function. Our algorithms may be categorized as a new class of ``descent-based" methods for function maximization on the nodes of a graph.
\end{abstract}

\section{Introduction}

Social media and big data analysis have led to an explosion of graph-structured datasets in  diverse domains such as neuroscience, economics, sensor networks, and databases. In many important instances, the object of interest may be expressed as a function defined on the vertices of the graph. For example, in social network analysis, it is useful to identify a node or nodes with a maximal centrality (or ``influence") score relative to the other nodes, for the purpose of performing strategic interventions~\cite{chen2012identifying, Jac10}. In electrical network modeling, the edges of a graph correspond to connections in an electrical grid, and the amount of power available at each node may similarly be represented using a graph-structured function~\cite{Vis13}. This has led to a variety of statistical studies in graph function estimation, including Laplacian smoothing~\cite{KonLaf02, SmoKon03} and graph trend filtering~\cite{WanEtal16}, when noisy measurements of graph characteristics may be available.

In this paper, we focus on the problem of obtaining the maximum value of a function defined on the nodes of a graph. We assume a setting where only local information is available, and any algorithm must make decisions based on graph information obtained by traversing the nodes of the graph in a sequential manner. For instance, such a setting may arise in web crawling~\cite{brautbar2010local}, respondent-driven sampling~\cite{Hec97}, or other demographic studies~\cite{BanEtal14}. Although basic properties of Markov chains imply that an unbiased random walk on the nodes of a connected graph will eventually visit every node---hence produce a function maximum---we wish to use information about the value of the function on neighboring nodes to guide the path of the algorithm, borrowing ideas from the extensive literature on function optimization in the continuous domain.

For continuous optimization, several algorithms such as gradient descent, stochastic gradient descent, stochastic gradient Langevin dynamics, and momentum-based descent algorithms~\cite{Ber99, WelTeh11, Nes83} provide guarantees for finding local and global optima when the function is sufficiently smooth. Common classes of smooth functions include Lipschitz functions, convex functions, and Polyak-\L{}ojasiewicz functions~\cite{Nes13}.  Most of these notions, however, do not have direct analogs in the discrete domain of graphs. For instance, there is no obvious analog of a convex function on a graph~\cite{LinEtal11, Hir17}. In this paper, we take a step toward systematically studying function optimization on graphs by considering a certain class of smooth graph functions, where the notion of smoothness is motivated by the theory of band-limited graph signals. Such topics constitute an area of study in graph signal-processing~\cite{ShuEtal13}, which has become popular within the signal processing community in recent years~\cite{anis2016efficient,puy2016random}. The algorithms we propose exploit the underlying graph structure to construct efficient algorithms for maximizing such smooth graph functions, producing a promising analogy to gradient descent algorithms for continuous functions. Two key features of all our algorithms are that they are iterative and local; i.e., the random walk relies only on the function values of neighboring nodes when deciding its next step. This is a direct analogy to first-order optimization methods in continuous domains.

Our proposed algorithms employ random walk methods based on Monte Carlo approaches, specifically the Metropolis-Hastings (MH) algorithm~\cite{christian1999monte}. An important idea is to reformulate the problem of graph function maximization as a problem involving sampling vertices with high function values. Such strategies have recently gained traction in continuous domain optimization, as well~\cite{Dal17, ZhaEtal17, RagEtal17}. Starting with a proposal probability density on the vertices of the graph that is related to the original function values in an appropriate manner, we construct random walk transition kernels that converge to the desired proposal distribution. Our approach differs significantly from comparable approaches in the literature in the following way: In constructing the MH random walk kernel, we exploit the graph-connectivity information implied by the graph Laplacian matrix to construct rapidly-mixing walks that converge to the function maximum. Through hitting time analysis, we provide theoretical guarantees for the speed of convergence of these algorithms.


The remainder of the paper is organized as follows: In Section~\ref{SecBackground}, we introduce the mathematical formulation of the graph maximization problem and the types of local algorithms and smooth functions we will consider. In Section~\ref{SecAlgorithms}, we provide a formal statement of our proposed algorithms, which are analyzed rigorously in Section~\ref{SecTheory}. In Section~\ref{SecExperiments}, we provide simulations demonstrating the behavior of our algorithms. We conclude with a list of interesting open questions in Section~\ref{SecDiscussion}.

\paragraph{\textbf{Notation:}} We write $\|\cdot\|_{TV}$ to denote the total variation norm and $\|\cdot\|_2$ to denote the Euclidean norm of a vector. For a matrix $A \in \mathbb{R}^{n \times n}$, we write $\diag(A)$ to denote the $n \times n$ matrix with diagonal entries equal to the components of $A$ and all other entries equal to zero. We write ${\left \lfloor{\cdot}\right \rfloor }$ to denote the floor operator.


\section{Background and Problem Setup}
\label{SecBackground}

Consider an unweighted, undirected graph $G$ with vertex set $V$ of size $n$ and symmetric adjacency matrix $W$. For a node $i \in V$, we write $d_i$ to denote the degree of $i$ and $N(i)$ to denote the neighborhood set of $i$. We write $d_{\max} = \max_{i \in V} d_i$. We refer to a function $f: V \to \mathbb{R}$ as a \emph{graph function} defined on the vertices of $G$, and sometimes write $f_i$ to denote the function value $f(i)$. Our goal is to design a local algorithm that finds the maximum value of an unknown graph function in an efficient manner.

\subsection{Local Search Algorithms}

Formally, we define a local search algorithm as a discrete time Markov process, where the states of the process are the $n$ nodes, and a Markov chain in state $X_t$ at time $t$ chooses the next state $X_{t+1}$ among the neighbors of $X_t$, with transition probability distribution defined as a function of the value $f(X_t)$ and the set of values $\{f(V_t): V_t \in N(X_t)\}$ at neighboring nodes in the graph.

We briefly mention some related work. Borgs et al.~\cite{BorEtal12} studied the problem of finding the first node in a preferential attachment network, where the algorithm is given access to all nodes in a frontier of radius $r$ around the set of previously queried nodes. Frieze and Pegden~\cite{FriPeg17} proposed an alternative algorithm of finding the root of a preferential attachment graph, where the algorithm is given access to a sorted list of the highest-degree neighbors of the current iterate. Brautbar and Kearns~\cite{brautbar2010local} studied algorithms for finding nodes with high degrees or high clustering coefficients using a ``jump and crawl" method, where a crawl step consists of randomly querying a neighbor of the existing iterate, and a jump step allows the algorithm to query a uniformly sampled node from the vertex set of the graph. Note, however, that our methods differ from the studies of Borgs et al.~\cite{BorEtal12} and Frieze and Pegden~\cite{FriPeg17} due to the fact that we are interested in different classes of graph functions besides the age function (indeed, the age of the nodes in the graph may not constitute a smooth function), and from the study of Brautbar and Kearns~\cite{brautbar2010local} due to the fact that successive steps may only visit immediate neighbors, rather than jumping to an entirely different node in the graph.

\subsection{Metropolis-Hastings Algorithm}

Many of our algorithms are based on versions of the Metropolis-Hastings Algorithm, which we review in this section. Given a state space $V$ and a target probability density $p_f(\cdot)$, the MH algorithm constructs a Markov chain with transition probability matrix $\mathbf{P}$ whose stationary distribution $\pi$ is the same as the target probability density (i.e., $\pi(i) = p_f(i) \ \forall \ i \in V$). The matrix $\mathbf{P}$ is defined in terms of a ``proposal'' distribution $\mathbf{Q}$, as follows:
\begin{equation}\label{eq_MHKernelFormula}
\mathbf{P}_{ij} = 
\begin{cases}
\mathbf{Q}_{ij} R(i,j), & j \neq i \\
1 - \sum_{j\neq i}\mathbf{P}_{ij}, & j = i,
\end{cases}
\end{equation}
where 
\begin{equation}\label{eq_MH_Rij_Formula}
R(i,j) = \min\left(1, \; \frac{p_f(j) \mathbf{Q}_{ji}}{p_f(i) \mathbf{Q}_{ij}} \right).
\end{equation}
The general idea behind the MH algorithm is that successive steps are guided by the proposal distribution, such that the move from node $i$ to node $j$ depends on the ratio $\frac{p_f(j)}{p_f(i)}$, but is also moderated by the proposal distribution according to the ratio $\frac{Q_{ji}}{Q_{ij}}$. Known results in probability theory guarantee the convergence of the MH algorithm to the target density~\cite{GriSti01}.

In addition to generating samples from a desired distribution, however, the MH algorithm is also a valuable tool in stochastic optimization~\cite{christian1999monte}. The key idea is that if a function $f$ with domain $V$ is encoded into an appropriate density $p_f$, such that larger values of $f$ correspond to larger values of $p_f$, the random walk generated by the MH algorithm will be more likely to visit nodes with larger density values, hence larger values of $f$.

The \emph{independent MH algorithm} corresponds to the choice of proposal distribution $\mathbf{Q}_{ij} = p_g(j)$, where $p_g$ is some probability measure over the state space $V$. In Section~\ref{SecAlgorithms}, we present a variety of algorithms for function maximization based on running the MH algorithm with different choices of proposal distributions and target densities, for which we may obtain provable guarantees for the rate of convergence.

\subsection{Smooth Graph Functions}

We wish to leverage information about local changes in the graph function across edges in the graph to guide the movement of an algorithm. Intuitively, this will lead to improvements in the rate of convergence of an algorithm to the graph maximum when the function satisfies certain smoothness properties, as a rough analog of the convergence results of gradient descent for smooth functions on a real domain. In order to make these notions rigorous, we now define the smooth class of graph functions to be analyzed in our paper, borrowing from the literature on graph signal processing.

The unnormalized graph Laplacian for a graph $G$ with adjacency matrix $W$ is given by $L = D - W$, where $D = \diag\left(\sum_{j=1}^n W_{ij}\right)$ is the degree matrix. Since the Laplacian is symmetric, we can calculate the eigendecomposition of $L$ as $U\Lambda U^T$, where $U \in \mathbb{R}^{n \times n}$ is an orthonormal matrix whose $i^\text{th}$ column $\bu_i$ is the eigenvector corresponding to the eigenvalue $\lambda_i$, and $\Lambda$ is a diagonal matrix with sorted eigenvalues $0 = \lambda_1 \leq \cdots \leq \lambda_n$. 

Since $U$ is unitary, the columns of $U$ form an orthogonal basis for any graph function $f$. In the basis of Laplacian eigenvectors, the function $f$ may be written as $f = U\hat f$, where $\hat f$ is defined as the graph Fourier transform $\hat f = U^T f$. The $i^\text{th}$ column $u_i$ of $U$ is the $i^\text{th}$ Fourier mode of the graph, corresponding to the $i^\text{th}$ smallest graph frequency, also equal to $\sqrt{\lambda_i}$ \cite{ShuEtal13}.

Intuitively, a ``smooth" graph function is a function that does not vary much across vertices that are connected by edges. This notion of smoothness is captured by the term $f^T L f$, as shown in the following identity \cite{von2007tutorial}:
\begin{equation}
f^T L f = \frac{1}{2}\sum_{i,j=1}^n w_{ij}(f_i-f_j)^2.
\end{equation}
For a smooth function, the sum of squared differences $(f_i-f_j)^2$ across all the connected vertices is small, implying that $f^T L f$ is small. Representing $f$ in the Fourier basis, we have
\begin{align*}
f^T L f = (\hat{f}^T U^T)U\Lambda U^T (U \hat f)
= \hat{f}^T \Lambda \hat f
= \sum_{i=1}^n \lambda_i \hat{f_i}^2.
\end{align*}
Since the $\lambda_i$'s are arranged in ascending order, setting the higher-order Fourier coefficients $\hat{f_i}$'s close to zero leads to a small value for $f^T L f$, implying that $f$ is smooth. We have the following definition:
\begin{definition}
A function $f: G \to \mathbb R$ is called \emph{$k$-smooth} if there exists a vector $\ba \in \mathbb{R}^k$ such that $f = U_k \ba$, where $U_k \in \mathbb{R}^{n \times k}$ is the matrix consisting of the first $k$ columns of $U$. 
\end{definition}

In general, any graph function $f$ can be decomposed as $f = f_{ks} + f_r$, where $f_{ks} = U_k U_k^T f$ is the $k$-smooth component and $f_r$ is the residual component. Note that $f_r \bot f_{ks}$, so
\begin{equation*}
\|f\|_2^2 = \|f_{ks}\|_2^2 + \|f_r\|_2^2.
\end{equation*}
In graph signal processing, a $k$-smooth function is viewed as a band-limited graph signal restricted to the lowest $k$ Fourier modes of the graph. The vector $\ba$ represents the projection of $f$ onto the lowest $k$ graph frequencies~\cite{puy2016random}. 

Denote $\delta_i$ to be the vector with $1$ in the $i^\text{th}$ coordinate and zeros elsewhere. We then have the following definition:
\begin{definition}
The term $\|U_k^T \delta_i\|_2$ is referred to as the {\em local cumulative coherence of order $k$} (LC-$k$) at node $i$ \cite{tremblay2016compressive}. The LC-$k$ quantity satisfies
\begin{equation}\label{eq_cumcoh}
\|U_k^T \delta_i\|_2 
= \frac{\|U_k^T \delta_i\|_2}{\|\delta_i\|_2}
=\frac{\|U_k^T \delta_i\|_2}{\|U^T \delta_i\|_2}.
\end{equation}
\end{definition}

From equation~\eqref{eq_cumcoh}, we see that $\|U_k^T \delta_i\|_2$ gives the proportion of the energy of the graph impulse function $\delta_i$ that is concentrated on the top $k$ Fourier modes \cite{puy2016random}. The LC-$k$ varies between $0$ and $1$.  Whereas we do not usually associate $\delta_i$ with smoothness, it may happen that for certain graph topologies, the value of LC-$k$ at node $i$ is close to 1, implying that $\delta_i$ is an approximately $k$-smooth function on the graph. Although LC-$k$ cannot be computed locally, methods exist for computing the value approximately, without having to perform an eigendecomposition of a potentially large graph Laplacian matrix~\cite{puy2016random,tremblay2016compressive}.
\begin{figure}
\centering
\label{fig:fig_smoothexample}
\caption{Examples of randomly generated smooth functions with $k=10$, $k=20$, and $k=30$ on a $32\times 32$ 2D grid graph of $n=1024$ nodes. The eigenvectors $u_{10}$, $u_{20}$ and $u_{30}$ are plotted in the first row. A higher intensity on a pixel indicates a higher positive value for the vector at that node. Note that as expected, functions corresponding to a smaller value of $k$ have fewer undulations.}
\includegraphics[scale=.8]{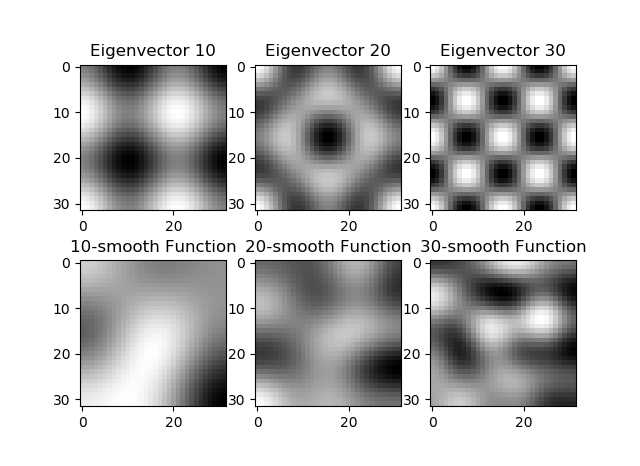}
\end{figure}


\section{Algorithms}
\label{SecAlgorithms}

We now describe the local algorithms that we will compare in this paper. A theoretical analysis of the relative convergence rates is provided in the next section.

\subsection{Vanilla Random Walk}

\begin{algorithm}[h!]
\caption{\texttt{Vanilla RW}}
\begin{algorithmic}
\label{alg_rw}
\REQUIRE Local access to graph $G$ (i.e., at every node, all its neighbors are accessible) and number of iterations $T$.
\begin{enumerate}
\item \textbf{initialize}:\\
\ \ \ \ $i \leftarrow \text{Unif}\{1,\ldots, n\}$ (Uniform sampling from $V$)\\
\ \ \ \ $i_{\max} \leftarrow i$\\
\ \ \ \ $f_{\max} \leftarrow f_i$ \\
\item
\textbf{repeat} for $T$ iterations:\\
\ \ \ \ $i \leftarrow \text{Unif}\{j: w_{ij}=1 \}$\\
\ \ \ \ \textbf{if} $f_i > f_{\max}$: \\
\ \ \ \ \ \ \ \ $f_{\max} \leftarrow f_i$\\
\ \ \ \ \ \ \ \ $i_{\max} \leftarrow i$\\
\textbf{return} $i_{\max}$
\end{enumerate}
\end{algorithmic}
\end{algorithm}

We begin by formally writing out the algorithm for an unbiased random walk on the vertices of the graph. The random walk proceeds around the nodes of the graph that are connected by edges. The transition kernel for the vanilla random walk is given by $P = D^{-1}W$. Thus, the transition probability from node $i$ to node $j$ is given by $P_{ij} = \frac{w_{ij}}{d_i}$, where $w_{ij}=1$ if an edge exists between $i$ and $j$, and $w_{ij} = 0$ otherwise.

The vanilla random walk is agnostic to the function $f$, and converges to a stationary distribution that is proportional to the degrees of the graph nodes  \cite{aldous2002reversible}. Note, however, that we cannot immediately conclude that such a function-agnostic random walk will not optimize the function efficiently---the class of $k$-smooth functions on graphs does not have a simple description, and it may happen that optimizers of such functions lie at high-degree nodes, which are precisely those nodes that attract the vanilla random walk.

\subsection{Exponentially-Weighted Walk}

The second algorithm involves biasing the choices of each move in the vanilla random walk according to the function values on the neighboring nodes. We choose an exponential weight function, so that the target probability density is defined by $p_f(i) \propto \exp\left(\gamma f_i\right)$. Starting with a proposal distribution $Q = D^{-1}W$, we use the MH method in equation~\eqref{eq_MHKernelFormula} to construct the transition kernel for the exponentially-weighted walk, as follows:
\begin{equation}\label{eq_ExpRW_Kernel}
\mathbf{P}_{ij} = 
\begin{cases}
\frac{1}{d_i} w_{ij} R(i,j) & j \neq i, \\
1 - \sum_{j\neq i}\mathbf{P}_{ij} & j = i,
\end{cases}
\end{equation}
where
\begin{align*}
R(i,j) 
= \min\left(1, \; \frac{p_f(j) (1/d_j)}{p_f(i) (1/d_i)} \right)
= \min\left(1, \; e^{\gamma(f_j - f_i)}\frac{d_i}{d_j} \right).
\end{align*}
Thus, the transition probability between any two nodes in the exponentially-weighted walk depends on the difference in function value $f_i-f_j$ and the degrees of the nodes. The parameter $\gamma$ determines the ``peakiness" of the target density $p_f$.

Unlike in the vanilla random walk, we assume that at every node, it is possible to access the function value and the degree of its neighbors. However, we do not make use of the smoothness constraint on the function. The parameter $\gamma$ may be viewed as controlling the greedy nature of this algorithm. When $\gamma = +\infty$, the random walk always ascends; i.e., if $f(j) < f(i)$, then $P(i,j) = 0$. However, such a greedy walk is susceptible to getting stuck at local maxima and failing to find the global maxima quickly. At the other extreme, when $\gamma = 0$, the exponential random walk is function-agnostic (as in the vanilla random walk), and converges to the uniform distribution over vertices.

\begin{algorithm}[h!]
\caption{\texttt{Exponential RW($\gamma$)}}
\begin{algorithmic}
\label{alg_rw_gamma}
\REQUIRE Local access to graph $G$ and function $f$ (i.e., at every node, all its neighbors and their function values are accessible), number of iterations $T$, and tuning parameter $\gamma$.
\begin{enumerate}
\item \textbf{initialize}:\\
\ \ \ \ $i \leftarrow \text{Unif}\{1,\ldots, n\}$ (Uniform sampling from $V$)\\
\ \ \ \ $i_{\max} \leftarrow i$\\
\ \ \ \ $f_{\max} \leftarrow f_i$ \\
\item
\textbf{repeat} for $T$ iterations:\\
\ \ \ \ Generate $j$ according to $\mathbf{P}_{ij}$ given in equation~\eqref{eq_ExpRW_Kernel}\\
\ \ \ \ $i \leftarrow j$\\
\ \ \ \ \textbf{if} $f_i > f_{\max}$: \\
\ \ \ \ \ \ \ \ $f_{\max} \leftarrow f_i$\\
\ \ \ \ \ \ \ \ $i_{\max} \leftarrow i$\\
\textbf{return} $i_{\max}$
\end{enumerate}
\end{algorithmic}
\end{algorithm}


\subsection{Graph Laplacian Walk}

For the third algorithm, we use a proposal distribution that is derived from the graph Laplacian, instead of the vanilla random walk proposal distribution $Q = D^{-1}W$ that we used for the exponentially-weighted walk.

Let $U_k$ represent the matrix containing $k$ eigenvectors corresponding to the smallest $k$ eigenvalues of the graph Laplacian. Let $\delta_i \in \{0,1\}^n$ represent the indicator vector for  node $i$. The proposal distribution for the graph Laplacian walk is given by
\begin{equation}\label{eq_LRW_Q'_ksmth}
\mathbf{Q^\prime}_{ij} = 
\frac{\|U_k^T \delta_j\|_2^2}{\sum_{j:w_{ij}=1} \|U_k^T \delta_j\|_2^2} w_{ij}.
\end{equation}
Note that the Laplacian-based proposal distribution $Q^\prime$ is indeed ``local," since $Q^\prime_{ij} = 0$ whenever $Q_{ij}=0$. 
For the Laplacian walk, we choose a target probability density according to $p_f(i) \propto f_i^2$. The squaring of the function restricts our analysis to maximizing positive functions.
Using the MH method in equation~\eqref{eq_MHKernelFormula}, we derive the transition kernel for graph Laplacian walk, as follows:
\begin{equation}\label{eq_LRW_KernelFormula}
\mathbf{P}_{ij} = 
\begin{cases}
Q_{ij}^\prime R(i,j) & j \neq i, \\
1 - \sum_{j\neq i}\mathbf{P}_{ij} & j = i,
\end{cases}
\end{equation}
where we set $R(i,j) = 0$ if $w_{ij} = 0$; and if $w_{ij} = 1$, we have
\begin{align}\label{eq_LRW_RijFormula}
R(i,j) 
= \min\left(1, \frac{p_f(j) \mathbf{Q^\prime}_{ji}}{p_f(i) \mathbf{Q^\prime}_{ij}} \right)
= \min\left(1,
\frac{f_j^2}{f_i^2}
\frac{\|U_k^T \delta_i\|_2^2}{\|U_k^T \delta_j\|_2^2}
\frac{\sum_{j:w_{ij}=1} \|U_k^T \delta_j\|_2^2}{\sum_{i:w_{ji}=1} \|U_k^T \delta_i\|_2^2}
\right).
\end{align}

\begin{algorithm}[h!]
\caption{\texttt{Laplacian RW(k)}}
\begin{algorithmic}
\label{alg_rw_k}
\REQUIRE Function smoothness $k$, local access to graph $G$, function $f$ and LC-$k$ (i.e., at every node, all its neighbors, their function values, and their LC-$k$ values are accessible/computable), and number of iterations $T$
\begin{enumerate}
\item \textbf{initialize}:\\
\ \ \ \ $i \leftarrow \text{Unif}\{1,\ldots, n\}$ (Uniform sampling from $V$)\\
\ \ \ \ $i_{\max} \leftarrow i$\\
\ \ \ \ $f_{\max} \leftarrow f_i$ \\
\item
\textbf{repeat} for $T$ iterations:\\
\ \ \ \ Generate $j$ according to $\mathbf{P}_{ij}$ given in equation~\eqref{eq_LRW_KernelFormula}\\
\ \ \ \ $i \leftarrow j$\\
\ \ \ \ \textbf{if} $f_i > f_{\max}$: \\
\ \ \ \ \ \ \ \ $f_{\max} \leftarrow f_i$\\
\ \ \ \ \ \ \ \ $i_{\max} \leftarrow i$\\
\textbf{return} $i_{\max}$
\end{enumerate}
\end{algorithmic}
\end{algorithm}

An important difference between the exponentially-weighted walk and the graph Laplacian walk is that the former algorithm is truly a local algorithm. Indeed, only the function values of the neighboring nodes are needed to determine the probability distribution for the next step of Algorithm~\ref{alg_rw_gamma}. This differs from Algorithm~\ref{alg_rw_k}, for which we assume knowledge of the graph Laplacian (or at least its top $k$ eigenvectors) in order to compute $\|U_k^T \delta_i\|_2^2$ for the neighboring nodes of every iterate.






We now describe a variant of Algorithm~\ref{alg_rw_k} applicable to approximately smooth functions. We begin with a definition:
\begin{definition}
A function $f$ is $\epsilon$-approximately $k$-smooth if for the decomposition 
$f = f_{ks} + f_r$ where $f_{ks} = U_k^T U_k f$, we have
\begin{align*}
|f_r(i)| \leq \epsilon \|f_{ks}\|_2,
\end{align*}
for all $i \in V$ and some $\epsilon>0$.
\end{definition}

For the class of $\epsilon$-approximate $k$-smooth functions, the proposal distribution for the graph Laplacian walk given in equation~\eqref{eq_LRW_Q'_ksmth} is replaced by
\begin{equation}\label{eq_LRW_Q'_epsAprxKsmth}
\mathbf{Q^\prime}_{ij} = 
\frac{\left(\|U_k^T \delta_j\|_2+\epsilon\right)^2}{\sum_{j:w_{ij}=1}  \left(\|U_k^T \delta_j\|_2 +\epsilon\right)^2} w_{ij}.
\end{equation}
The transition kernel for the Laplacian walk will remain the same as in equation~\eqref{eq_LRW_KernelFormula}, with the modified $Q_{ij}^\prime$ in equation~\eqref{eq_LRW_Q'_epsAprxKsmth}.


\section{Theoretical Analysis}
\label{SecTheory}

In this section, we provide results concerning the convergence rate of the local search algorithms. For both of our algorithms, we derive bounds on (i) the total variation distance between the probability distribution at time $t$ and the stationary distribution of the Markov chain; and (ii) bounds on the hitting time of the algorithm in expectation and in high probability. A practical consequence of the total variation bounds is that if we run the corresponding local algorithms for sufficiently many steps, we are guaranteed that taking the maximizer of the empirically constructed distribution will be provably close to the function maximizer. The hitting time bounds may be interpreted as a bound on the amount of time needed to first visit a maximum---thus, if we halt the algorithms after a prescribed number of steps and choose the node with the largest function value so far, we are guaranteed to obtain a maximum with high probability. In fact, if we are given slightly more knowledge (i.e., that we have located a function maximizer upon visiting it for the first time), the theorems provide stronger guarantees for a variant of the local algorithms that halt once they identify a maximum.



\subsection{Exponentially-Weighted Walk}

In this section, we analyze the exponentially-weighted random walk in Algorithm \ref{alg_rw_gamma}.

\subsubsection{Convergence}

In the first theorem, we show the convergence of the exponentially-weighted random walk to the target density $p_f$ in total variation norm.

\begin{theorem}[Convergence of Algorithm \ref{alg_rw_gamma}]\label{thm_convExpRW}
For a connected graph $G$ with diameter $r$, and for a graph function $f$ that is nonzero on all vertices, the rate of convergence of the random walk proposed in Algorithm \ref{alg_rw_k} to its stationary distribution $p_f$ is given by
\begin{equation}\label{eq_conv_alg_rw_gamma}
|| \mathbf{P}^t_{i*} - p_f||_{TV} 
\leq 
\Bigg( 1 - \frac{\delta_f^{r-1}}{(d_{max}\Delta_f)^r}\Bigg)
^{\left \lfloor{\frac{t}{r}}\right \rfloor },  \ \ \forall \ i \in V,
\end{equation}
where $\mathbf{P}^t_{i*}$ is the $i^\text{th}$ row of $\mathbf{P}^t$ (the transition probability matrix after $t$ steps), $\delta_f = \min_i p_f(i)$, $\Delta_f = \max_i p_f(i)$.
\end{theorem}

\begin{proof}
For any two vertices $i$ and $j$ connected by an edge, we have
\begin{align*}
P_{ij} 
&= Q_{ij}\min\Big( 1, \frac{p_f(j) Q_{ji}}{p_f(i) Q_{ij}} \Big)\\
&= \frac{1}{d_i}\min\Big( 1, \frac{p_f(j) d_i}{p_f(i) d_j} \Big)\\
&= \min\Big( \frac{1}{d_i}, \frac{p_f(j)}{p_f(i) d_j} \Big)\\
&= p_f(j)\min\Big( \frac{1}{d_ip_f(j)}, \frac{1}{d_j p_f(i)} \Big)\\
&\geq \frac{1}{d_{max}\Delta_f}p_f(j) \\
& \geq \frac{\delta_f}{d_{max}\Delta_f},
\end{align*}
For any two vertices $i$ and $j$ that are not connected by an edge, it is possible to find at most $r-1$ vertices that connect $i$ to $j$, since the diameter of the graph is $r$. Hence, in all cases,
\begin{align*}
\mathbf{P}^r_{ij}
\geq \left(\frac{\delta_f}{d_{max}\Delta_f}\right)^{r-1}\frac{1}{d_{max}\Delta_f} p_f(j)
= \frac{\delta_f^{r-1}}{(d_{max}\Delta_f)^r} p_f(j),
\end{align*}
implying that
\begin{align*}
|| \mathbf{P}^{rs+m}_{i*} - p_f||_{TV}
= \theta^s || (Q^{s} P^m)_{i*} - p_f||_{TV}
\leq \theta^s.
\end{align*}
We now adopt a proof technique from Theorem 4.9 in Levin and Peres~\cite{levin2017markov}. Define $P_f \in \mathbb{R}^{n \times n}$ such that $P_f(i,j) = p_f(j)$. We may write
\[
P^r = (1-\theta)P_f + \theta B,
\]
where $\theta = 1 - \frac{\delta_f^{r-1}}{(d_{max}\Delta_f)^r}$ and $Q$ is a stochastic matrix. 
By induction, we will prove that 
\begin{equation}\label{eq_p^t1}
P^{rs} = (1-\theta^{s})P_f + \theta^s B^{s}.
\end{equation}
Assuming equation~\eqref{eq_p^t1} is true for $s$, we have
\begin{align*}
P^{r(s+1)} = P^{rs} P^r 
&= (1-\theta^s)P_fP^r + \theta^s B^s((1-\theta)P_f + \theta B)\\
&= (1-\theta^s)P_f + (1-\theta)\theta^sB^s P_f + \theta^{s+1}B^{s+1}\\
&= (1-\theta^{s+1})P_f + \theta^{s+1} B^{s+1},
\end{align*}
since $P_f P^r = P_f B = P_f P_f  = P_f$.
Multiplying equation~\eqref{eq_p^t1} by $P^m$, where $0 \leq m < r$, we then have
\begin{align*}
P^{rs+m} = (1-\theta^{s})P_f P^m + \theta^s B^{s} P^m
= P_f + \theta^s(B^{s} P^m - P_f),
\end{align*}
implying that
\begin{align*}
P^{rs+m} - P_f = \theta^s(B^{s} P^m - P_f).
\end{align*}
Hence,
\begin{align*}
\|\mathbf{P}^{rs+m}_{i*} - p_f\|_{TV}
= \theta^s \|(B^{s} P^m)_{i*} - p_f\|_{TV}
\leq \theta^s.
\end{align*}
Taking $t = rs+m$, where 
$s = {\left \lfloor{\frac{t}{r}}\right \rfloor }$, completes the proof.
\end{proof}

\subsubsection{Results on Hitting Times}

In the next theorem, we prove an upper bound on the expected number of steps it takes for Algorithm \ref{alg_rw_gamma} to reach the function maximum.

\begin{theorem}[Expected hitting time]\label{Thm_exphit_ExpRW}
For Algorithm \ref{alg_rw_gamma}, let $v_t \in V$ be the state at time $t$, and define $T_{hit} = \min\{t\geq 0: f_{v_t}=f_{max}\}$ to be the number steps it takes for Algorithm \ref{alg_rw_gamma} to reach the function maximum. Let $f_{min}$ denote the minimum value of the function $f$. Then the expected value of the hitting time is bounded by
\begin{equation}\label{eq_hittime_main}
\mathbb{E} T_{hit} \leq d_{max}^r e^{\gamma(r-1)(f_{max}-f_{min})}.
\end{equation}
\end{theorem}

\begin{proof}
From the proof of Theorem \ref{alg_rw_k}, we have
\begin{align}\label{eq_p^t-pf}
P^{t} - P_f = \theta^s\left(Q^{s} P^m - P_f\right),
\end{align}
for $\theta = 1 - \frac{\delta_f^{r-1}}{(d_{max}\Delta_f)^r}$, and $t = rs+m$ and 
$s = {\left \lfloor{\frac{t}{r}}\right \rfloor }$.
Define hitting time for state $i$ as $T_i = \min\{t\geq 0: v_t=i\}$. Let $\mathbb{E}_i T_j$ denote the expected hitting time for state $j$ given that the MH Markov chain starts at state $i$.
From Lemma 2.12 in Aldous~\cite{aldous2002reversible}, we have
\[
p_f(j) \mathbb{E}_i T_j = Z_{jj} - Z_{ij},
\]
where $Z_{ij} = \sum_{t=0}^\infty (P_{ij}^t - p_f(j))$. Using equation~\eqref{eq_p^t-pf}, we then have
\begin{align*}
p_f(j) \mathbb{E}_i T_j = \sum_{t=0}^\infty (P_{jj}^t - P_{ij}^t) 
 = \sum_{t=0}^\infty \theta^t \left((Q^{s} P^m)_{jj}^t - (Q^{s} P^m)_{ij}^t\right)
& \leq \sum_{t=0}^\infty \theta^t 
= \frac{1}{1-\theta}
 = \frac{(d_{max}\Delta_f)^r}{\delta_f^{r-1}}.
\end{align*}
Hence,
\begin{equation*}
\mathbb{E}_i T_j \leq \frac{(d_{max}\Delta_f)^r}{\delta_f^{r-1} p_f(j)},
\end{equation*}
and
\begin{align*}
\mathbb{E}_i T_{hit} 
\leq \frac{(d_{max}\Delta_f)^r}{\delta_f^{r-1} \Delta_f} 
= d_{max}^r \left( \frac{\Delta_f}{\delta_f} \right)^{r-1} 
= d_{max}^r e^{\gamma(r-1)\left(f_{max}-f_{min}\right)}.
\end{align*}
\end{proof}

We now derive a high-probability bound on the hitting time.

\begin{theorem}[High-probability bound on hitting time]\label{thm_highprob_ExpRW}
For any initial distribution $\mu$ and any $s>0$, we have
\begin{equation}\label{eq_ubhit}
\mathbb{P}_{\mu}[T_{hit} > t] \leq \left(\frac{t^{\star}_{hit}}{s}\right)^{\lfloor t/s \rfloor },
\end{equation}
where $t^{\star}_{hit} = d_{max}^r e^{\gamma(r-1)(f_{max}-f_{min})}$ and $t>0$.
\end{theorem}

\begin{proof}
For any integer $m\geq 1$, we have
\begin{align*}
\mathbb{P}_{\mu}[T_{hit} > ms \mid T_{hit} > (m-1)s] &
 \leq \max_j \mathbb{P}_j (T_{hit} > s)
 \leq \frac{\mathbb{E}_j(T_{hit})}{s} = \frac{t_{hit}^\star}{s}.
\end{align*}
By induction on $m$, we obtain 
$\mathbb{P}_{\mu}[T_{hit} > ms] \leq \left(\frac{t^{\star}_{hit}}{s}\right)^m$, which implies inequality~\eqref{eq_ubhit}.
\end{proof}

\subsection{Graph Laplacian Walk}

We begin by proving a few lemmas concerning the proposal distribution $Q^\prime$ used in Algorithm~\ref{alg_rw_k} for $k$-smooth and $\epsilon$-approximately $k$-smooth functions. We establish the following envelope condition:
\begin{equation}\label{eq_env}
\frac{1}{M}p_f(j) \le Q^\prime_{ij},
\end{equation}
for an appropriate constant $M$, when $w_{ij}=1$.

\begin{lemma}[Dominance for $k$-smooth graph functions]\label{lem_dom_ksmth}
Suppose $f$ is $k$-smooth. For the proposal distribution $Q^\prime$ given in equation~\eqref{eq_LRW_Q'_ksmth}, the envelope condition~\eqref{eq_env} holds with $M=k$.
\end{lemma}

\begin{proof}
Since $f$ is $k$-smooth, we have $f = U_k \ba$ for some $\ba \in \mathbb{R}^k$. It follows that $f_i = \delta_i^T U_k \ba = \langle U_k^T\delta_i, \ba\rangle$. Hence, 
\begin{equation*}
f_i^2 = \langle U_k^T\delta_i, \ba\rangle^2 \leq \|U_k^T \delta_i\|_2^2 \|\ba\|_2^2,
\end{equation*}
implying that
\begin{align*}
\sum_{i \in V} f_i^2 = \|f\|_2^2 = \|U_k^T \ba\|_2^2 = \|\ba\|_2^2.
\end{align*}
Thus, for $j\neq i$, we have
\begin{align*}
p_f(i) = \frac{f_j^2}{\sum_{i \in V} f_i^2}
\leq \|U_k^T \delta_i\|_2^2.
\end{align*}
Using equation~\eqref{eq_LRW_Q'_ksmth}, for $w_{ij}=1$, we then have
\begin{align*}
Q_{ij}^\prime 
= \frac{\|U_k^T \delta_j\|_2^2}{\sum_{j:w_{ij}=1} \|U_k^T \delta_j\|_2^2}
\geq \frac{p_f(j)}{\sum_{j:w_{ij}=1} \|U_k^T \delta_j\|_2^2}
\geq \frac{p_f(j)}{\sum_{j=1}^n \|U_k^T \delta_j\|_2^2}
= \frac{1}{k}p_f(j).
\end{align*}
\end{proof}

\begin{lemma}[Dominance for $\epsilon$-approximate $k$-smooth graph functions]\label{lem_dom_apprx}
Suppose $f$ is $\epsilon$-approximate $k$-smooth. For the proposal distribution $Q^\prime$ given in equation~\eqref{eq_LRW_Q'_epsAprxKsmth}, the envelope condition~\eqref{eq_env} holds with
\begin{equation}
M = k + 2k\sqrt{n}\epsilon + n\epsilon^2.
\end{equation}
\end{lemma}

\begin{proof}
Let $f_{ks} = U_k^T\ba$ for some $\ba \in \mathbb{R}^k$.Then
\begin{align*}
f(i)^2 
= (f_{ks}(i) + f_r(i))^2 
\leq \left(\langle U_k^T\delta_i, \ba\rangle + \epsilon \|f_{ks}\|_2\right)^2
\leq \left(\|U_k^T \delta_i\|_2 \|\ba\|_2+\epsilon \|\ba\|_2\right)^2,
\end{align*}
implying that
\[
\sum_{i \in V} f(i)^2 
= \|f\|_2^2 = \|f_{ks}\|_2^2 + \|f_r\|_2^2.
\geq ||\ba||_2^2
\]
Thus, for $j\neq i$, we have
\begin{align*}
p_f(i) = \frac{f_j^2}{\sum_{i \in V} f_i^2}
\leq \left(\|U_k^T \delta_i\|_2+\epsilon\right)^2.
\end{align*}
From equation~\eqref{eq_LRW_Q'_epsAprxKsmth}, for $j\neq i$ and $w_{ij}=1$, we have
\begin{align*}
Q_{ij}^\prime
& =\frac{\left(\|U_k^T \delta_j\|_2+\epsilon\right)^2}{\sum_{j:w_{ij}=1} \left(\|U_k^T \delta_j\|_2 +\epsilon\right)^2}
\geq \frac{p_f(i)}{{\sum_{j=1}^n \left(\|U_k^T \delta_j\|_2+\epsilon\right)^2}} \\
& \geq \frac{p_f(i)^2}{k+n\epsilon^2 
+ 2\epsilon\left(\sqrt{n}\sum_{i\in V} \|U_k^T \delta_i\|_2^2\right)}
=\frac{p_f(i)^2}{M}.
\end{align*}
\end{proof}

\subsubsection{Convergence}

The first theorem concerns convergence in TV distance.

\begin{theorem}[Convergence of Algorithm \ref{alg_rw_k}]
For a connected graph $G$ with diameter $r$, and for a $k$-smooth graph function $f$ (either exact or $\epsilon$-approximate) that is nonzero on all the vertices,
the rate of convergence of the random walk proposed in Algorithm~\ref{alg_rw_k} to its stationary distribution $p_f$ is given by
\begin{equation}\label{eq_conv_alg_rw_k}
\|\mathbf{P}^t_{i*} - p_f\|_{TV} 
\leq 
\Bigg( 1 - \frac{\delta_f^{r-1}}{M^r}\Bigg)
^{\left \lfloor{\frac{t}{r}}\right \rfloor },  \ \ \forall \ i \in V,
\end{equation}
where $\mathbf{P}^t_{i*}$ is the $i$th row of $\mathbf{P}^t$, $\delta_f = \min_i p_f(i)$, and $M$ is the dominance constant established in Lemmas~\ref{lem_dom_ksmth} and~\ref{lem_dom_apprx}.
\end{theorem}

\begin{proof}
From equations~\eqref{eq_LRW_KernelFormula} and~\eqref{eq_LRW_RijFormula}, for any two vertices $i$ and $j$ connected by an edge, we have
\begin{align*}\label{eq_lwconvproof1}
P_{ij} 
=Q^\prime_{ij}
\min\left(1, \frac{p_f(j) Q^\prime_{ji}}{p_f(i) Q^\prime_{ij}} \right)
= p_f(j)
\min\left(\frac{Q^\prime_{ij}}{p_f(j)}, \frac{Q^\prime_{ji}}{p_f(i)} \right)
\geq \frac{1}{M} p_f(j),
\end{align*}
where the last inequality follows from Lemmas~\ref{lem_dom_ksmth} and~\ref{lem_dom_apprx}.

For any two vertices $i$ and $j$ that are not connected by an edge, it is possible to find at most $r-1$ vertices that connect $i$ to $j$, since the diameter of the graph is $r$. This leads to the following inequality:
\begin{align*}
\mathbf{P}^r_{ij}
\geq \Big(\frac{\delta_f}{M}\Big)^{r-1}\frac{1}{M} p_f(j)
= \frac{\delta_f^{r-1}}{M^r} p_f(j).
\end{align*}
Similar to proof of Theorem~\ref{thm_convExpRW}, define $P_f \in \mathbb{R}^{n \times n}$ such that $P_f(i,j) = p_f(j)$. We then write
\[
P^r = (1-\theta)P_f + \theta B,
\]
where $\theta = 1 - \frac{\delta_f^{r-1}}{M^r}$ and $B$ is a stochastic matrix. Following a similar analysis as that in Theorem~\ref{thm_convExpRW} with the modified $\theta$ completes the proof.
\end{proof}

\subsubsection{Results on Hitting Times}

In the next theorem, we prove an upper bound on the expected number of steps it takes for Algorithm \ref{alg_rw_k} to reach the function maximum.

\begin{theorem}[Expected hitting time]\label{Thm_exphit_LW}
For Algorithm \ref{alg_rw_k}, suppose $v_t \in V$ be the state at time $t$, and define $T_{hit} = \min\{t\geq 0: f_{v_t}=f_{max}\}$. The expected value of the hitting time is bounded by
\begin{equation}\label{eq_hittime_main1}
\mathbb{E} T_{hit} \leq \frac{(M||f||^2)^r}{f_{\max}^2 f_{\min}^{2(r-1)}}.
\end{equation}
\end{theorem}

\begin{proof}
Following the analysis of Theorem \ref{Thm_exphit_ExpRW} with $\theta = 1 - \frac{\delta_f^{r-1}}{M^r}$, we have
\begin{align*}
p_f(j)\mathbb{E}_i T_j
\leq \frac{1}{1-\theta}
= \frac{M^r}{\delta_f^{r-1}},
\end{align*}
where $\mathbb{E}_i T_j$ is the expected hitting time to reach node $j$, starting from node $i$. For $T_{hit}$, we substitute $j = i_{\max}$, so that
\begin{equation*}
p_f(i_{\max}) = \max_{i}p_f(i) = \Delta_f.
\end{equation*}
Hence,
\begin{align*}
\mathbb{E}T_{hit} 
\leq \frac{M^r}{\delta_f^{r-1}\Delta_f}
= \frac{(M||f||^2)^r}{f_{\max}^2 f_{\min}^{2(r-1)}}.
\end{align*}
\end{proof}

The high-probability bound on expected time would be same as that in Theorem~\ref{thm_highprob_ExpRW}, with
\begin{equation*}
t_{hit}^\star = \frac{(M||f||^2)^r}{f_{\max}^2 f_{\min}^{2(r-1)}}.
\end{equation*}


\section{Experiments}
\label{SecExperiments}

In this section, we consider various graph topologies and approximately $k$-smooth functions defined on the nodes of the graphs and compare the algorithms described in Section~\ref{SecAlgorithms}. 
\begin{figure}
\centering
\label{fig:fig_example_20smth}
\caption{An example run of the algorithms on a $20$-smooth function on a $32\times 32$ 2D grid graph. White pixels in rows 2--5 indicate nodes that have been visited by the algorithm. White pixels in the function plot denote high function values. Notice that the Laplacian RW is very effective at traversing regions of the graph where the function takes large values. The exponential RW with $\gamma = 1$ is also effective at reaching the function peaks located at the coordinates (15,10). The vanilla random walk covers a large area but without any preferred direction, as expected.}
\includegraphics[scale=0.6]{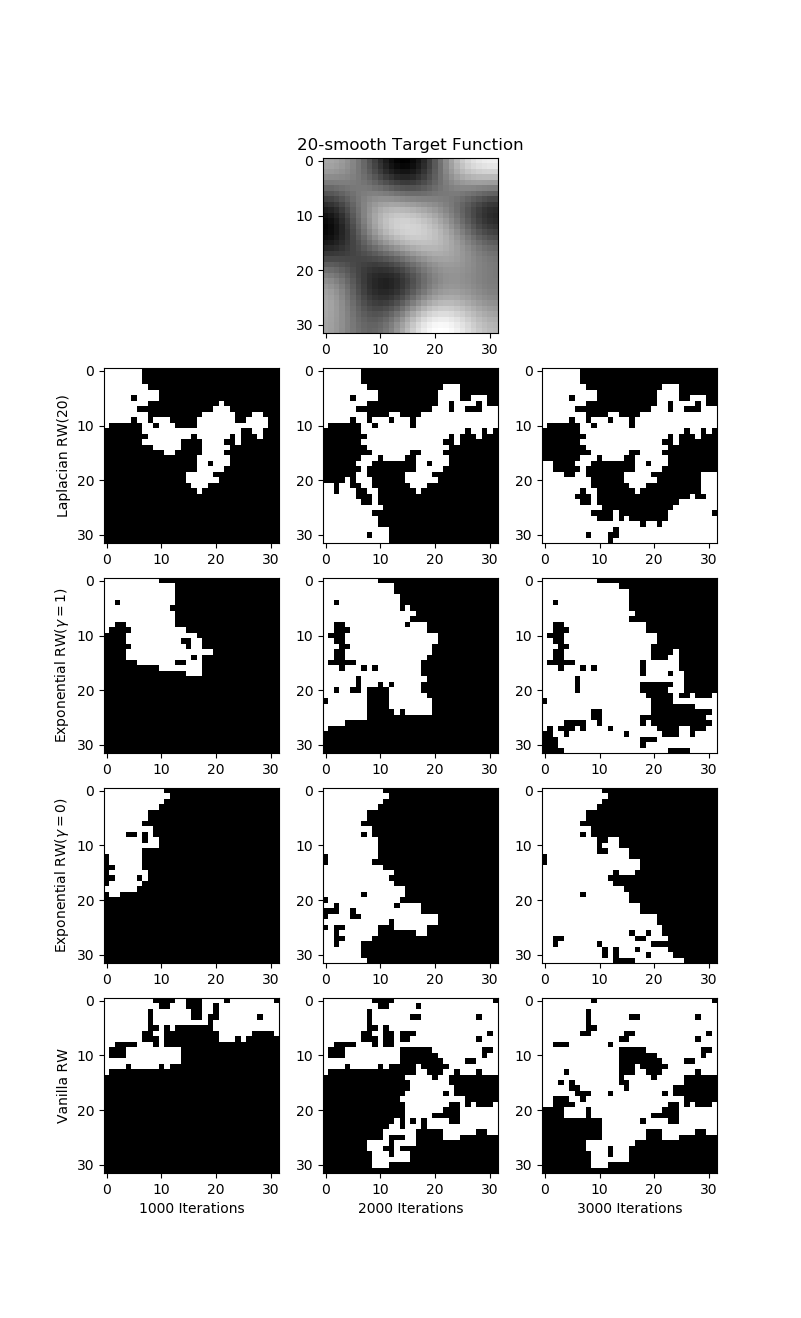}
\end{figure}
We simulated Algorithms 1, 2 and 3 on the following graph models: (i) 2D Grid graph, (ii) Erd\H{o}s-Renyi (ER) graph, (ii) and Barabasi-Albert (BA) graph. For the 2D grid graph, we ran the algorithms on a $32 \times 32$ grid of $n=1024$ nodes. For the ER graph, we generated a random graph with $n=1000$ nodes with the probability of an edge between any two nodes being $p = \frac{1.1\log n}{n}$, and discarded the graphs with isolated nodes so that the entire graph formed one connected component. For the BA graph, we generated a random graph with $n=1000$ nodes starting with a seed set of $m=3$ nodes.

For each of these graphs, we generated random smooth functions for varying values of $k$. For each value of $k$, we sampled $\ba$ from a $k\times 1$ vector of standard normal random variables, and constructed the smooth function as $f = U_k\ba$, where $U\in \mathbb{R}^{n \times k}$ is the matrix containing the top $k$ eigenvectors of the graph Laplacian matrix. The functions were then made nonnegative by lifting the function value at every node by the minimum. Since the constant vector is the first eigenvector of the graph Laplacian, this process of adding a constant vector does not affect the smoothness of the function.

The algorithms were run till the random walk ``hit" the maximum, or up to a maximum of $10,000$ steps. The hitting times were then averaged over $100$ iterations for each value of $k$, and for $10$ randomly chosen functions (i.e., a total of $1000$ iterations for each value of $k$). For Algorithm~\ref{alg_rw_gamma}, we ran the experiments with two different values of $\gamma$. The case $\gamma=0$ corresponds to a function-agnostic random walk that converges to a uniform stationary distribution across all the nodes. The case $\gamma = 1$ corresponds to a moderately greedy random walk that converges to a stationary distribution proportional to $e^{\gamma f}$. A high value of $\gamma$ would reduce Algorithm \ref{alg_rw_gamma} to a `greedy' walk that gets stuck in local minima, which is not desirable. For instance, in our experiments, $\gamma=10$ led to very poor performance due the algorithm often being unable to move out of local minima. 

Figures~\ref{fig:fig_k_ER}, \ref{fig:fig_k_BA}, and~\ref{fig:fig_k_2D}, report the results. As can be seen, the Laplacian walk of Algorithm~\ref{alg_rw_k} consistently outperformed all other algorithms for all the three classes of graphs we considered. It is also worth noting that changing the comparison criterion from hitting the global maximum to hitting the top $1\%$ of nodes did not change the performance of the algorithms in our simulations.

\begin{figure}
\centering
\caption{Comparison of hitting times for Erd\H{o}s-Renyi graphs as a function of smoothness}
\includegraphics[scale=0.7]{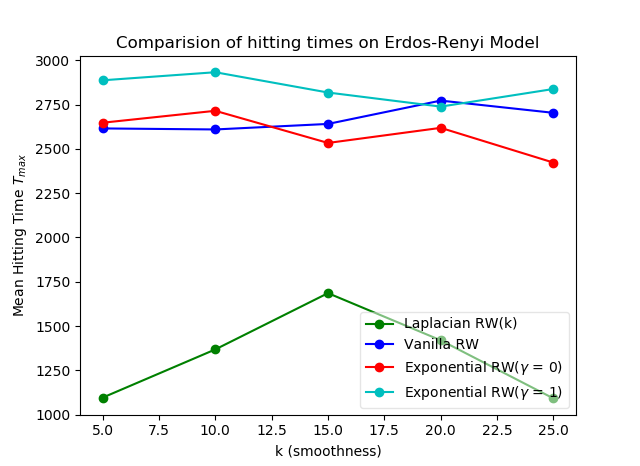}
\label{fig:fig_k_ER}
\end{figure}

\begin{figure}
\centering
\caption{Comparison of hitting times for Barabasi-Albert graphs as a function of smoothness}
\includegraphics[scale=0.7]{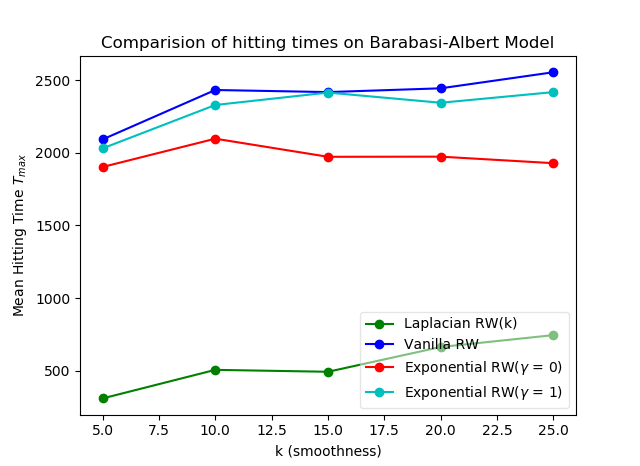}
\label{fig:fig_k_BA}
\end{figure}

\begin{figure}
\centering
\caption{Comparison of hitting times for 2D grid graphs as a function of smoothness}
\includegraphics[scale=0.7]{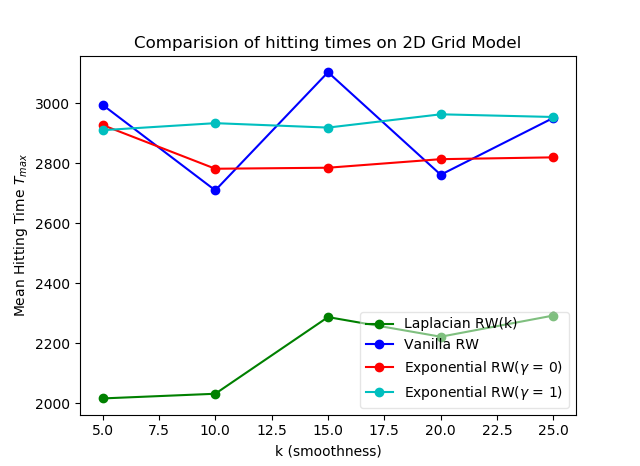}
\label{fig:fig_k_2D}
\end{figure}


\section{Discussion}
\label{SecDiscussion}

We have presented two new algorithms for graph function maximization based on local movements around the nodes of a graph. Our first algorithm concerns an exponentially weighted random graph, and our second algorithm uses information about the spectrum of the graph Laplacian matrix. We have provided theoretical results concerning the rate of convergence of our algorithms in terms of total variation distance and hitting time when the graph function belongs to a certain smoothness class.

The algorithms we have studied in this paper only involve local movements along edges, from one node to an adjoining neighbor. However, one might imagine variants of the algorithms that allow ``jump" movements that reinitialize the walk at a randomly drawn vertex, either uniformly selected from all the nodes or all the nodes previously explored in previous time steps. It would be interesting to see how incorporating the option of a jump move might affect the convergence analysis of the algorithms.

Another question to explore is the tightness of our bounds on rates of convergence of the algorithms. Our results are derived based on a seemingly coarse analysis, and it would be interesting to see if it is possible to find worst-case graphs and classes of smooth functions for which one can also derive lower bounds for the rate of convergence of any local algorithm. This also gives rise to the important related question of the landscape of local and global optima of a $k$-smooth function, where the smoothness is defined in terms of eigenvalues of the graph Laplacian. This appears to be a challenging problem even for $k=2$ and for Erd\H{o}s-Renyi random graphs. Ideally, we would also like to be able to translate the bounds on total variation distance and hitting time into precise recommendations regarding the number of iterates required for each of our algorithms to locate a maximum.

A final question concerns improving the convergence rate of local algorithms when the next iterate is allowed to depend on the values of the function in a neighborhood of radius $r$ around the current iterate. The algorithms described in this paper are limited to the case when $r=1$. However, if the local algorithm were given information about a larger neighborhood at each step, perhaps it would be possible to devise an analog of higher-order descent algorithms, which are known to exhibit faster rates of convergence in continuous optimization settings.

\bibliographystyle{plain}
\bibliography{ref}

\end{document}